\documentclass[journal]{IEEEtran}
\usepackage{amsthm, graphicx, cite, picinpar, amsmath, url, flushend, colortbl, soul, multirow, pifont, color, alltt, enumerate, siunitx, breakurl, pbox, balance}
\usepackage{algorithm}
\usepackage{euscript,amsmath}
\usepackage{amssymb}

\usepackage[font=small,skip=0pt]{caption}
\newtheorem{lemma}{Lemma}
\newtheorem{remark}{Remark}
\newtheorem{theorem}{Theorem} 

\begin{document}
\title{Asynchronous Periodic Distributed Event-Triggered Frequency Control of Microgrids}

\author{
%	\vskip 1em
%	Joined work
	Keywan Mohammadi, \emph{Student Member,~IEEE},
	Elnaz Azizi, \emph{Graduate Student Member,~IEEE},
	Mohammad-Taghi Hamidi-Beheshti, \emph{Member,~IEEE},
    Ali Bidram, \emph{Member,~IEEE},
    and	Sadegh Bolouki, \emph{Member,~IEEE}
	
\thanks{K.~ Mohammadi, E.~Azizi, M.~Beheshti, and S.~Bolouki are with Department of Electrical and Computer Engineering, Tarbiat Modares University, Tehran, Iran.e-mails:(keywan.mohammadi@modares.ac.ir, e.azizi@modares.ac.ir, mbehesht@modares.ac.ir, bolouki@modares.ac.ir)

A.~Bidram is with Department of Electrical and Computer Engineering, The University of New Mexico, Albuquerque, New Mexico. email:(bidram@unm.edu)}
}

\maketitle
\begin{abstract}
In this paper, we introduce a distributed secondary frequency control scheme for an islanded ac microgrid under event-triggered communication.
An integral type event-triggered mechanism is proposed by which each distributed generator (DG) asynchronously and periodically checks its triggering
condition and determines whether to update its control inputs and broadcast its states to neighboring DGs. In contrast to existing event-triggered
strategies on secondary control of microgrids, under the proposed sampled-data based event-triggered mechanism, DGs need not be synchronized to a common
clock and each individual DG checks its triggering condition periodically, relying on its own clock. Furthermore, the proposed method efficiently reduces
communication and computation complexity. We provide sufficient conditions under which all DGs' frequencies asymptotically converge to the common reference
frequency value. Finally, effectiveness of our proposed method is verified by simulating different scenarios on a well-established islanded ac microgrid benchmark in the MATLAB/Simulink environment.
\end{abstract}
\begin{IEEEkeywords}
Asynchronous event-triggered control,\;
distributed secondary control,\;
frequency restoration,\;
islanded microgrid,\; 
multi-agent System.
\end{IEEEkeywords}

%\markboth{IEEE TRANSACTIONS ON INDUSTRIAL INFORMATICS}%
%{}

\definecolor{limegreen}{rgb}{0.2, 0.8, 0.2}
\definecolor{forestgreen}{rgb}{0.13, 0.55, 0.13}
\definecolor{greenhtml}{rgb}{0.0, 0.5, 0.0}

\section{Introduction}
    Emerging distributed energy resources have shaped a new structure in power distribution networks, paving the way for creation of the microgrid concept \cite{mousavi2018autonomous}. In normal operation, microgrids are connected to the main grid and their voltage and frequency are imposed by the upstream grid. A microgrid can get disconnected from the main grid and go to the autonomous mode. Despite the advantages of microgrids in enhancing the power system's flexibility, they present some technical challenges such as control and power management issues. Hence, in microgrids, a hierarchical control scheme is tasked to ensure reliable performance in the face of probable challenges \cite{lasseter2002microgrids}.
    
    Decentralized primary controller, which is located at the innermost layer of the hierarchical structure, deals with fast dynamics and stability of the microgrid system \cite{bidram2012hierarchical, guerrero2010hierarchical}. However, the primary control level causes deviations of voltage and frequency from their nominal rating. Therefore, to compensate the steady-state deviations, an outer control layer, namely secondary controller, can be applied. Restoring frequency and voltage magnitudes caused by the primary controller is the main objective of the secondary control level.
    
    Early research on the secondary control of microgrids mitigated steady-state deviations of voltage and frequency in a centralized manner \cite{bidram2012hierarchical, guerrero2010hierarchical, mehrizi2010potential}. However, due to the heavy communicational burden on the central controller and high sensitivity of the network to centralized architectures may reduce the system's reliability. Therefore, distributed cooperative control strategies with sparse and robust communication networks, became appropriate alternatives for the secondary controller design \cite{bidram2013secondary}.

    On the other hand, distributed cooperation and coordination within networked systems have become the focal point of research in a wide variety of scientific and engineering problems \cite{qin2016recent,cao2012overview,gulzar2018multi}. Considering the problem of frequency and voltage synchronization in microgrids as a leader-follow consensus problem, the secondary control design can be conducted based upon distributed coordination theory in multi-agent systems. In much of the research in this field, considering continuous time communication between DGs as an assumption is evident. However, discrete sample-data interaction is more realistic for data exchange in communication networks. Furthermore, in practice, frequently gathering information and updating control actions exhaust communication and computation capabilities of DGs' digital tools. This has led to the emergence of event-triggered control strategies as sound alternatives to sampled-data techniques \cite{henningsson2008sporadic,lunze2010state}. 
    
    Almost all the recent efforts in event-triggered secondary control of microgrid systems have been done under the assumption of DGs' capability of continuously or periodically but synchronously evaluating the triggering condition during the process. Existing synchronous periodic event-triggered techniques need a globally synchronized clock, according to which all DGs evaluate their event conditions, update their control signals, and broadcast their states to other DGs. However, due to the large scales of microgrid systems, synchronization to a common clock seems not to be always reasonable and imposes physical limitations. This paper proposes a novel distributed event-triggered secondary control in order to tackle this challenge.
    
    \subsection{Related Work}
    
    The very first attempts to design secondary controller utilizing distributed cooperative control theory were \cite{bidram2013secondary} and \cite{bidram2014multiobjective}. In these articles, the nonlinear and heterogeneous dynamics of the DGs are transformed to the linearized dynamics using feedback linearization method. Consequently, the voltage and frequency restoration problem resembles a linear distributed tracking problem which has been widely studied in the multi-agent systems literature. In \cite{xu2018optimal}, authors introduce a finite-time framework for the distributed secondary controller, by which the frequency regulation and active power sharing are well achieved while a decoupled design for voltage regulation and reactive power sharing at different time scales with the frequency controller is enabled. Considering noisy measurements, a distributed noise-resilient secondary control is proposed in \cite{dehkordi2018distributed}, in which a mean-square average-consensus protocol has been employed to regulate voltage and frequency in case of corrupted communication channels. Time delay effects on the secondary control layer is thoroughly addressed in reference \cite{ahumada2015secondary}, showing that the model predictive controller has more robustness in case of time delays. In \cite{amoateng2017adaptive}, the model based distributed controllers are designed firstly and then adaptive neural networks are utilized to approximate the uncertain/unknown dynamics of the microgrid system.
    
    Event-triggered techniques have been investigated in distributed secondary control of microgrid systems in \cite{ding2018distributed,chen2017secondary,fan2016distributed,zhou2019distributed}. In \cite{ding2018distributed},  a distributed secondary active power sharing and frequency control, based on a sample-based event-triggered communication strategy, is proposed that effectively reduces the communication complexity. In \cite{chen2017secondary}, utilizing event-triggered secondary voltage and frequency control, fair sharing of both active and reactive powers between power sources are investigated. Authors in \cite{fan2016distributed} utilized an event-triggered mechanism for active and reactive power sharing of microgrids. Considering uncertainties a distributed $H_\infty$ consensus approach with an event-triggered communication scheme is presented in \cite{zhou2019distributed}.
\subsection{Contributions}

    To the best of authors' knowledge, no research has been dedicated to the asynchronous event-triggered secondary control problem of microgrid systems yet. In our proposed method, each DG is equipped with its own clock and may have different event-checking instants from the rest of the system. Cyber network problems under asynchronous communication are clearly more complicated than those under synchronous communication, as the latter set of problems can be viewed as special cases of the former. In order to fill this gap, this paper investigates the cooperative secondary control problem of ac microgrid systems based on asynchronous periodic event-triggered strategy, bringing model one step closer to reality. This paper has the following salient contributions that, to the best of our knowledge have not been exploited yet:

    \begin{itemize}
    
        \item A distributed secondary frequency scheme using an event-triggered mechanism is proposed. It is demonstrated that the proposed mechanism is able to achieve a nearly identical frequency regulation while reducing the rate of communication and computation.
        
        \item Compared with existing event-triggered mechanisms on secondary control of microgrids, this paper is the first to propose an event-checking mechanism which is capable coping with the asynchronous event-checking behavior of DGs.
        
        \item From practical perspective, unlike traditional synchronized event-triggered mechanisms, our model setup does not require DGs to be coordinated to a global synchronized clock. Therefore, implementation of our proposed mechanism is more practical and efficient than existing GPS clock based mechanisms.
        
    \end{itemize}
    The remainder of this paper is organized as follows. Section II presents the preliminaries of graph theory, while Section III provides the dynamical modeling of an autonomous microgrid. The proposed secondary frequency control schemes is presented in section IV. In Section V, the effectiveness of the proposed secondary control method is validated on a microgrid test system using MATLAB/Simulink software environment. Finally, this paper is concluded in VI , where future directions of this research are stated. 

\section{Preliminaries on Graph Theory}

    In this work, we consider a network of DGs whose communication topology is represented by a weighted, directed , simple graph $\mathcal{G} =(V,E,A)$, in which $V=\{v_1,v_2,\dots,v_n\}$ is the set of nodes, each representing a DG, $E \subset V\times V$  represents the edge sets, each representing a directed communication channel from a DG to another, and $A=[a_{ij}]\in R^{n \times n}$ is the generalized adjacency matrix formed by edge weights of the graph that are all assumed non-negative. Concretely, an edge from DG $i$ to DG $j$ exists if there is a communication channel from DG $i$ to DG $j$, i.e., DG $i$ is able to send data to DG $j$. We notice that this channel of communication can be as well inferred from the value of $a_{ij}$. More precisely, DG $i$ is able to send data to DG $j$ if and only if $a_{ij} > 0$. The graph Laplacian matrix of $\mathcal{G}$ is defined as $L=[l_{ij}] \in \mathbb{R}^{(n \times n)}$, in which $l_{ii}=\sum_{j\neq i}a_{ij}$  and $l_{ij}=-a_{ij}$. A directed path from $v_i$ to $v_j$ is a sequence of edges, expressed as $\{(v_i,v_k),(v_k,v_l),...,(v_m,v_j)\}$. A directed graph is called $strongly$ $connected$ if there exists a directed path from any node to any other node \cite{ren2007information}.
    
    \begin{lemma} {\normalfont \cite{xiao2009finite}}
    If $\mathcal{G}$ is a strongly connected directed graph with Laplacian matrix $L$, there exists a vector $w=[w_1, w_2, \dots, w_n]$ with all positive elements such that $wL=0$. Furthermore, defining $W = diag(w_1, w_2, \dots, w_n)$, the matrix $WL+L^TW$  is semi definite.
    \end{lemma}

\section{Dynamical Modeling of an Autonomous Microgrid}

        A microgrid is a complex dynamical system consisting of physical layers, control layers, and cyber infrastructures. An inverter-based DG as the main building block of the microgrid system is depicted in Fig.~\ref{Inverter_based DG}. The large signal dynamical model of each DG is represented on its own direct and quadrature (d-q) reference frame. For constructing the model of the entire system, the reference frame of one DG is assigned as the common frame with the rotating frequency of  $\omega_{com}$. The dynamics of other DGs must then be translated to this common one, i.e., loads and lines dynamics are represented on the common frame. Details on transformation equations are provided in \cite{pogaku2007modeling}.
        \begin{figure}[t]
            \centering            \includegraphics[scale=0.64]{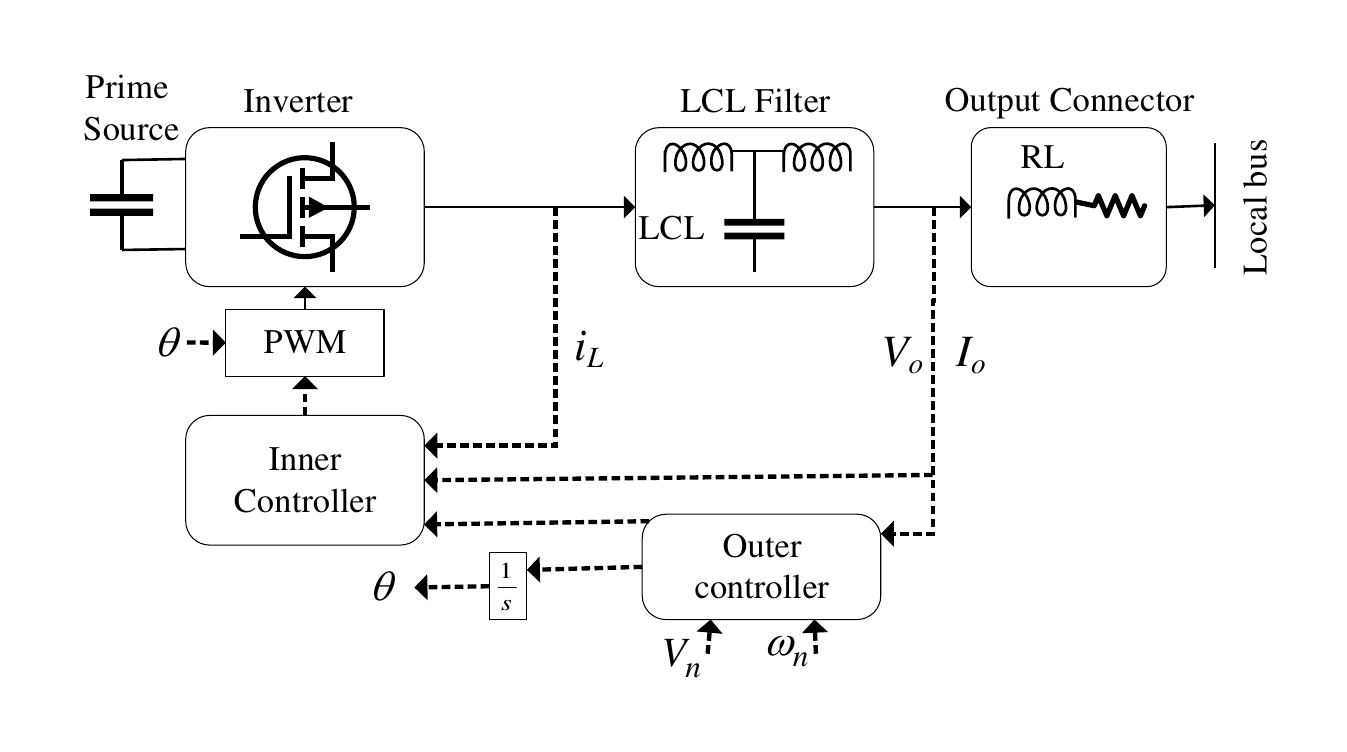}
            \caption{Block diagram of an inverter-based DG.}
            \label{Inverter_based DG}
        \end{figure}
        
        Compensating voltage and frequency deviations can be defined as a steady-state error elimination problem. Thus, for the secondary controller design, neglecting fast dynamics of inner loops due to their poor effects on the steady-state performance of the microgrid system could be permissible \cite{bidram2013secondary,rasheduzzaman2015reduced}. Accordingly, they are in this case removed from the modeling equations \cite{bidram2013secondary}. Then, the algebraic equations of the droop controller are written as \cite{bidram2013secondary}
        \begin{equation} \label{Droop}
            \begin{cases}
                \omega_i = \omega_{ni} -m_{pi}P_i\vspace{.05in}\\     v^*_{odi} = V_{ni}-n_{Qi}Q_i \vspace{.05in}\\   v^*_{oqi}=0
            \end{cases}
        \end{equation}
        where $\omega_{ni}$ and $V_{ni}$ are respectively the nominal setpoint of rotating frequency and the output voltage provided by the secondary controller. Internal control loops voltage references are $v_{odi}^*$ and $v_{oqi}^*$. The operating frequency of $i^{\rm{th}}$ inverter bridge is $\omega_{i}$. The LC filter and output connector differential equations are expressed in the following ($v_{id}=v_{od}^*$) \cite{pogaku2007modeling}:
        \begin{equation} \label{Eq2}
            \begin{cases}
                \dot{i}_{ldi} = \frac{-r_{fi}}{L_{fi}}i_{ldi} + \omega_ii_{lqi} +\frac{1}{L_{fi}}v_{idi}-\frac{1}{L_{fi}}v_{odi}\vspace{.05in}\\ \dot{i}_{lqi}=\frac{-r_{fi}}{L_{fi}}i_{lqi} -\omega_ii_{ldi}+\frac{1}{L_{fi}}v_{iqi}-\frac{1}{L_{fi}}v_{oqi} \vspace{.05in}\\
                \dot{v}_{odi}=\omega_iv_{oqi}+\frac{1}{C_{fi}}i_{ldi}-\frac{1}{C_{fi}}i_{odi}\vspace{.05in}\\
               \dot{v}_{oqi}=-\omega_iv_{odi}+\frac{1}{C_{fi}}i_{lqi}-\frac{1}{C_{fi}}i_{oqi}
               \vspace{.05in}\\
               
               \dot{i}_{odi} = \frac{-r_{ci}}{L_{ci}}i_{odi}+\omega_ii_{oqi}+\frac{1}{L_{ci}}v_{odi}-\frac{1}{L_{ci}}v_{bdi}\vspace{.05in}\\
               
               \dot{i}_{oqi} = \frac{-r_{ci}}{L_{ci}}i_{oqi}-\omega_ii_{odi}+\frac{1}{L_{ci}}v_{oqi}-\frac{1}{L_{ci}}v_{bqi}
            \end{cases}
        \end{equation}
        Equations \eqref{Droop} and \eqref{Eq2} can be written in a multi-input multi-output (MIMO) nonlinear compact form as
        \begin{equation*}
            \begin{cases}
                \dot{x}_i=f_i(x_i)+g_{i1}(x_i)u_{i1}+g_{i2}(x_i)u_{i2}+k_i(x_i)D_i \vspace{.05in}\\
                y_i=h_i(x_i)
            \end{cases}
        \end{equation*}
        where $x_i=[i_{ldi}, i_{lqi}, v_{odi}, v_{oqi}, i_{odi}, i_{oqi}]^T$  consists of the direct and quadratic components of $i_{li}$, $v_{oi}$ and $i_{oi}$, $u_i=[u_{i1}, u_{i2}]^T=[\omega_{ni}, V_{ni}]^T$, $D_i=[v_{bdi}, v_{bqi}]^T$ are control and disturbance inputs, respectively, and $y_i=[y_{i1}, y_{i2}]^T=[v_{odi}, \omega_i ]^T$ is formed by the output voltage and frequency.

%%%%%%%%%%%%%%%%%%%%%%%%%%%%%%%%%%%%%%%%%%%%%%%%%%%%%%%%%%%%%%%%%%%
\section{Asynchronous Periodic Event-Triggered Secondary Frequency Control}

    We consider an islanded ac microgrid with $N$ DGs, each of which contains a primary source, voltage source inverter (VSI), an LC filter, and an output connector. A basic control framework of the primary control layer is shown in Fig.~\ref{Inverter_based DG}. In what follows, first the problem statement is presented, and then, the proposed asynchronous periodic event-triggered distributed secondary frequency controller is developed. 
    
\subsection{Problem Statement}

    The control issue which is considered in this paper is to regulate the operating frequency of an islanded ac microgrid based on distributed cooperation control of multi-agent systems. This controller selects proper control input $\omega_{ni}$ in \eqref{Droop} to restore the operating frequency of DGs, $\omega_i$, to their reference value, $\omega_{ref}$. We herein assume that the communication framework among DGs is described by a strongly connected directed graph. Recalling the frequency droop characteristic in \eqref{Droop}, one can establish a relation between operating frequency $\omega_i$ and the control input, $\omega_{ni}$ as
        \begin{equation} \label{relation}
           \omega_{i}=\omega_{ni}-m_{pi}P_i.
        \end{equation}  
    Differentiating both sides of \eqref{relation} and defining an auxiliary control input $u_{\omega i}$, one has
        \begin{equation} \label{auxilary}
            \dot{\omega}_i=  \dot{\omega}_{ni}-m_{pi}\dot{P}_i=u_{wi},
        \end{equation}    
    \noindent where $u_{\omega i} \in \mathbb{R} $.  
    
    In this work, we consider the problem of asynchronous behavior of the communication network in microgrid systems. It is assumed that the local secondary frequency controller of each DG, described by \eqref{auxilary}, samples its desired states at fixed period times, relying on its own clock. This is an arguably expected behavior for real-world  multi-agent systems like microgrids. Since such systems cover a large-scale geographical area, synchronization to a global synchronized clock needs GPS based infrastructures which is not always convenient. In such circumstances, even though DGs have similar fixed sampling period times, they may start to sample their states at different time instants, resulting in asynchronous communication throughout the process. In the particular case of event-triggered controller design, the system's asynchronous communication behavior leads to asynchronous event-checking time instants. Therefore, we aim to design an event-triggered control mechanism that, beside their capability of reducing communication and computation complexity, are also be able to handle asynchronous communication within the network. Accordingly, as illustrated in Fig.~\ref{Event-detection}, in spite of DGs fixed periodic event-checking samplings, they may have different event-checking instants with respect to the rest of the microgrid system.

    Let the common event-checking period be denoted by $h$ and DGs' starting times $t_0^1, t_0^2,...,t_0^N$ belong to the time interval $[0,h)$. Thus, each DG $i$ checks its triggering condition at discrete times $t_0^i, t_1^i, \ldots $, where $t_k^i=t_0^i+kh$, $\forall k > 0$. Let the sequence $\left(t_{(0)}^i, t_{(1)}^i, \ldots\right)$ denotes the event instants of DG $i$, which is a subsequence of event-checking instants $\left(t_0^i, t_1^i,\ldots\right)$. Then, 
    \begin{equation} \label{latest broadcast}
             \hat{\omega}_i(t)=\omega_i(t^i_{(k)}), t^i_{(k)}\le\; t < t^i_{(k+1)},~k=0,1,2,...
    \end{equation}
    \noindent defines the latest broadcast operating frequency of DG $i$ at any given time $t$. In other words, $\hat{\omega}_i$ is a piecewise constant function that only changes value at event times.
    \begin{figure}
        \centering
        \includegraphics[scale=0.55]{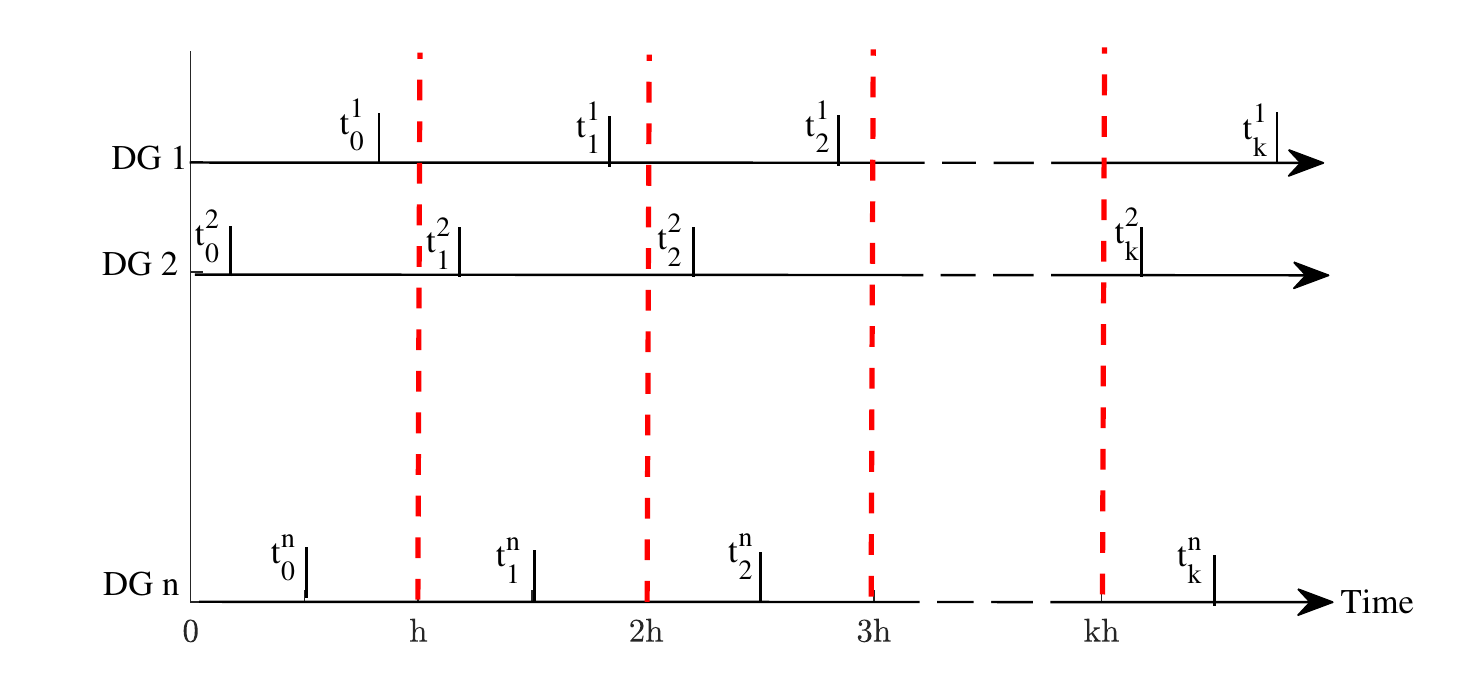}
        \caption{Event-checking time instants of the DGs.}
        \label{Event-detection}
    \end{figure}
\subsection{Control Design}

    Given the assumptions and initialization described above, we investigate the frequency restoration problem for the asynchronous distributed system in \eqref{auxilary} under periodic event-triggered control. The auxiliary control input is given as
        \begin{equation} \label{auxilary control}
          u_{\omega_i}(t)=-c_{\omega} e_{\omega_i}(t),~t\ge h,
        \end{equation}
    \noindent where $c_\omega \in R$ is the frequency control gain which adjusts the convergence speed and $e_{\omega i}$ is the following neighborhood tracking error:
        \begin{equation} \label{Error}
            e_{\omega_i}(t)=\sum_{j\in N_i} \tilde{a}_{ij}(\hat{\omega}_i(t)-\hat{\omega}_j(t))+\tilde{g}_i(\hat{\omega}_i(t)-\omega_{ref}),~t\ge h,
        \end{equation}
    \noindent in which $\tilde{a}_{ij} \ge 0$ is the edge weight of the communication graph and indicates the communication strength between two connected DG $i$ and DG $j$.  The pinning gain denoted as $\tilde{g}_i$; $\tilde{g}_i > 0$ if and only if the $i^{th}$ DG is connected to the reference. The time invariant and constant state of the leader (reference) node is denoted as $\omega_{ref}$. 
    From \eqref{Error}, the global tracking error vector can be defined as
        \begin{equation} \label{T error}
            e_{\omega}(t)=(\tilde{L}+\tilde{G})(\hat{\omega}(t)-{\omega}_{ref}\mathbf{1}_n),~t\ge\; h,
        \end{equation}
    where,
        \begin{equation*}
            e_{\omega} (t)=[ e_{\omega 1}(t),   e_{\omega 2}(t),  \dots,  e_{\omega n}(t)]^T,
        \end{equation*}
        \begin{equation*}
           \hat{\omega}(t)=[\hat{\omega}_{1}(t),    \hat{\omega}_{2}(t),  \dots,  \hat{\omega}_{n}(t)]^T,
        \end{equation*}
    and $\mathbf{1}_n$ is the vector of all ones with size $n$. Let $\tilde{L}$ be the graph Laplacian of $\mathcal{G}=(V,E,A)$ and $\tilde{G} =\text{diag}(\tilde{g}_1,\tilde{g}_2, \dots, \tilde{g}_n)$  be the pining gain vector.
   
   \begin{remark}
\normalfont In \eqref{auxilary control}, updating control protocol $u_{\omega i}$ depends on the neighborhood tracking error $e_{\omega i}$, which itself depends on the latest updated information of the piecewise constant functions $\hat{\omega} (t)$. Therefore, $u_{\omega i}$ is updated at both its own event times and those of its neighbors. Furthermore, it is worth mentioning that since the overall system does not synchronously start to activate and broadcast data, we assume for each DG $i$ that $\tilde{a}_{ij}=0$, $\forall t < max\left(t_0^i,t_0^j\right)$.
      \end{remark}     
    Combining the controller gain $c_\omega$ with the term $\tilde{L}+\tilde{G}$ and defining $\hat{\delta}(t) = \left(\hat{\omega} (t)-\omega_{ref} 1_n \right)$  as the global disagreement vector, given the global tracking error defined in \eqref{T error}, the closed-loop model for the linear system \eqref{auxilary} is represented by
        \begin{equation} \label{closed loop}
         \dot{\omega}(t)=-(L+G)\hat{\delta}(t),~t\ge \; h.
        \end{equation}
    Motivated by the findings of \cite{wang2017consensus}, we suggest the following event-triggered condition: 
    \begin{equation} \label{Proposed C}
        \begin{aligned}
          &|\omega_i(t^i_{(k)}+ph)-\omega_i(t^i_{(k)})| > \vspace{.05in}\\
            &\sigma_{\omega} \sqrt{\frac{\int_{t^i_{(k)}+(p-1)h}^{t^i_{(k)}+ph} (L_i\hat{\omega}(z)+g_i(\hat{\omega}_i(z)-\omega_{ref}))^2\, dz}{h}},
        \end{aligned}
\end{equation}
    \noindent where $\sigma_{\omega}$ is a positive scalar, $L_i$ is the $i^{th}$ row of the graph Laplacian $L$, and $t_{(k)}^i+ph$ is the $p^{th}$ event-checking instant after the latest event at $t_{(k)}^i$ for DG $i$. It should be clear that decreasing $\sigma_{\omega}$ enhances the chance of event occurring for each DG $i$ at any given time.
    
    \begin{remark}
    \normalfont The main purpose of the control mechanism based on the event condition \eqref{Proposed C} is to reduce the communication cost and the number of control updates while guaranteeing restoration for the operating frequency of the system. Once the triggering condition is met, the current state of DG $i$ is sampled and broadcast to its own controller as well as its neighbors. It should be noted that we do not consider time delays in communications in this work.
    \end{remark}
    
    Before stating our main theorem, we define
            \begin{equation} \label{lamda}
                \lambda = max_{\|\omega\|_2=1}\frac{\omega^T(L+G)^TW(L+G)\omega}{\omega^TW(L+G)\omega},
            \end{equation}
    where  $(L+G)^TW(L+G)$ and $W(L+G)+(L+G)^TW$ are positive-definite matrices. Moreover, denoting $A=L+G$ and recalling that $\hat{\delta}(t)=(\hat{\omega}(t)-\omega_{ref}\mathbf{1}_n)$ and $\delta(t)=({\omega}(t)-\omega_{ref}\mathbf{1}_n)$, we have
        \begin{equation}\label{eqsb2}
         \dot{\delta}(t)=-A\hat{\delta}(t).
        \end{equation}
    We now state a sufficient condition under which the proposed control law leads to the convergence to consensus of all DGs' frequencies.
    
    \begin{theorem}\label{main theorem}
    Given a strongly connected directed graph among the DGs, let the asynchronous system \eqref{closed loop} be driven by the event-triggering mechanism \eqref{Proposed C}. Then, the operating frequency terms $\omega_{i}$, $1 \leq i \leq n$, converge to $\omega_{ref}$ if the event checking period $h$ and the positive parameter $\sigma_{\omega}$  satisfy the inequality
        \begin{equation} \label{Eq16}
            \frac{h}{2}+\sigma_{\omega} < \frac{1}{\lambda}.
        \end{equation}   
    \end{theorem}

    \begin{proof} To prove the theorem, we consider the following candidate Lyapunov function
        \begin{equation*}
        V(t)=\frac{1}{2}\delta(t)^TW\delta(t).
        \end{equation*}   
    From \eqref{eqsb2}, differentiating $V(t)$ results in
        \begin{equation*}
           \dot{V}(t)=-\delta(t)^TWA\hat{\delta}(t).
        \end{equation*}   
    Now, following similar lines of argument as in the proof of Theorem 1 of \cite{wang2017consensus}, we first arrive at
    \begin{equation*}
        \lim_{k \to \infty} \hat{\delta}(t_k )^T WA \hat{\delta}(t_k )=0.
    \end{equation*}
    Then, recalling \eqref{lamda}, we have
    \begin{equation*}
            \hat{\delta}(t_k)^TA^T WA \hat{\delta}(t_k ) \le \lambda \hat{\delta}(t_k)^T WA \hat{\delta}(t_k),   
    \end{equation*}
    which immediately implies that $\lim_{k \to \infty}A\hat{\delta}(t_k)=0$.
    This, together with \eqref{latest broadcast}, result in
    \begin{equation} \label{lim}
        \lim_{t \to \infty}A\hat{\delta}(t)=0, 
    \end{equation}
    or equivalently,
    \begin{equation}\label{eqsb5}
        \lim_{t \to \infty} \dot{\delta}(t) = 0.
    \end{equation} 
    Recalling the event-triggering condition \eqref{Proposed C}, we notice that the following inequality holds for any $k$:
    \begin{equation} \label{inequality}
        \begin{aligned}
         & |\omega_i(t_k^i)-\hat{\omega_i}(t_k^i)|\le \sigma_{\omega}\sqrt{\cfrac{\int_{t_{k-1}^i}^{t_{k}^i}(A_i\hat{\delta}(z))^2dz}{h}},\\
         & ~~~~~~~~~~~~~~~~~~~~~~~i = 1,\dots, n.
        \end{aligned}
    \end{equation}
    More precisely, if an event occurs at $t_k^i$, then $\hat{\omega}_i(t_k^i)=\omega_i(t_k^i)$, which means that \eqref{inequality} holds. If no event occurs at $t_k^i$, then $\hat{\omega}_i(t_k^i)=\hat{\omega}_i(t_{k-1}^i)$, meaning that the event-triggering condition \eqref{Proposed C} is not satisfied, implying that \eqref{inequality} holds. Relations \eqref{lim} and \eqref{inequality} imply that for any $i$,
    \begin{align}\label{eqsb3}
        \lim_{k \to \infty} (\omega_i(t_k^i)-\hat{\omega}_i(t_k^i)) = 0.
    \end{align}
    Since for any $t$, there exists a positive integer $k_t$  such that $t \in [t_{k_t}^i,t_{k_{t}+1}^i]$, we can write
    %\begin{equation}
        \begin{align}\label{eqsb4}
                &\lim_{t \to \infty} (\omega_i(t)-\hat{\omega}_i(t))\vspace{.05in}\nonumber\\
                & = \lim_{t \to \infty} (\omega_i(t)- \omega_i(t_{k_t}^i) +\omega_i(t_{k_t}^i)-\hat{\omega}_i(t))\vspace{.05in}\nonumber\\
                & = \lim_{t \to \infty} ((\omega_i(t)-\omega_{ref})-(\omega_i(t_{k_t}^i)-\omega_{ref})\vspace{.05in}\nonumber\\ & ~~~+\omega_i(t_{k_t}^i)-\hat{\omega}_i(t_{k_t}^i))\vspace{.05in}\nonumber\\
                & = \lim_{t \to \infty} (\delta_i(t)- \delta_i(t_{k_t}^i) +\omega_i(t_{k_t}^i)-\hat{\omega}_i(t_{k_t}^i))\vspace{.05in}\nonumber\\
                & = \lim_{t \to \infty} \left(\int_{t_{k_t}^i}^{t}\dot{\delta}_i(z)dz +\omega_i(t_{k_t}^i)-\hat{\omega_i}(t_{k_t}^i)\right)=0.
        \end{align}
   We note that the last equality in \eqref{eqsb4} is deduced from \eqref{eqsb5} and \eqref{eqsb3}. From \eqref{eqsb4}, we conclude that $ \lim_{t \to \infty} (\omega(t)-\hat{\omega}(t)) = 0$, and consequently, $\lim_{t \to \infty} (\delta(t)-\hat{\delta}(t)) = 0$. Thus,  $\lim_{t\to \infty} A\delta(t) = 0$. Hence,
        \begin{equation*}
        \begin{aligned}
                 &\lim_{t\to \infty} \delta(t)^TWA\delta(t) \\
                 &= \lim_{t\to \infty} \delta(t)^T\left(\frac{WL+L^TW}{2} +WG\right)\delta(t)=0,
        \end{aligned}
        \end{equation*}  
    which implies that
    \[
        \lim_{t\to \infty} \delta(t) = 0.
    \]
    Thus we finally have
        \begin{equation}
            \lim_{t\to \infty} \omega(t) = \lim_{t\to \infty} A(\delta(t)+\omega_{ref}1_n)=\omega_{ref}\mathbf{1}_n,    
        \end{equation}
    which completes the proof.
    \end{proof}
    According to \eqref{auxilary} and \eqref{auxilary control}, $\omega_{ni}$ is written as
        \begin{equation*}
         \omega_{ni}=\int(u_{wi}+m_{pi}\dot{P}_i) \, dt.
        \end{equation*}
    Although the secondary frequency controller eliminates frequency steady state deviations, it may lead to worse active power sharing compared to the primary controller. However, one expects that once the secondary frequency control is applied, the control system will still be able to provide the same power sharing pattern guaranteed by the primary controller \cite{bidram2013secondary}. Applying the primary droop controller, the following equality is then satisfied:
        \begin{equation} \label{equality}
             m_{P_{1}}{P}_{1}=\dots =  m_{P_{n}}{P}_{n},
        \end{equation} 
    where $m_{P_i}$ denotes the active power rating of each DG $i$.
    
    Similar to the primary controller, the secondary frequency controller should guarantee \eqref{equality}. In order to achieve this requirement, an extra control input should be defined. Differentiating \eqref{equality} and defining a control input, the power sharing problem is transformed to the consensus problem of first-order multi-agent systems
        \begin{equation*}
            \begin{cases}
            m_{P_{1}}\dot{P}_{1} = u_{p_1}\vspace{.05in}\\
            m_{P_{2}}\dot{P}_{1} = u_{p_2}\vspace{.05in}\\
            \dots\vspace{.05in}\\
            m_{P_{N}}\dot{P}_1 = u_{p_N}
            \end{cases}
        \end{equation*} 
    Given a strongly connected communication network topology among DGs, the auxiliary control input $u_{Pi}$ is established as 
        \begin{equation*}
            u_{P_i}(t)=c_{P_i}e_{p_i}(t),
        \end{equation*}      
    where $c_{Pi}$ is the active power control gain and $e_{Pi}$ is the following neighboring tracking problem:
        \begin{equation} \label{neighboring T}
          e_{P_i}(t)=\sum_{j\in{N_i}}\tilde{a}_{i_j}(m_{P_i}\hat{P}_i(t)-m_{P_j}\hat{P}_j(t)),~t\ge h.
        \end{equation}      
    Since the power sharing problem is a consensus problem, DGs must reach a non-prescribed agreement according to their power ratings. So compared to \eqref{Error}, there is no external reference input in \eqref{neighboring T}. 
    We now present the following event-triggering mechanism
    \begin{equation}\label{eq32}
        \begin{aligned}
            &|m_{pi}P_i(t^i_{(k)}+ph)-m_{pi}P_i(t_{(k)}^i)|\vspace{.05in}\\
            &>\sigma_{P}\sqrt{\frac{\int_{t_{(k)}^i+(p-1)h}^{t_{(k)}^i+ph}(\hat{\Omega}(z))^2 \,dz}{h}},
        \end{aligned}
    \end{equation} 
    \noindent where 
    \begin{equation*}
        \hat{\Omega}(z)=\sum_{j \in N}a_{ij}(m_{pi}\hat{P}_i(z)-m_{pj}\hat{P}_j(z)).
    \end{equation*}
    Using the same procedure as in Theorem \ref{main theorem}, we can prove that the DGs active power asymptotically converge to a common non-prescribed value. Then, the control input $\omega_{ni}$ is written as 
    \begin{equation*}
        \omega_{ni}= \int (u_{\omega i}+u_{Pi})dt.
    \end{equation*}
    Fig.~\ref{Secondary_control} illustrates the block diagram of the proposed secondary frequency and active power controllers.
        \begin{figure*} [t]
        \centering
        \includegraphics[scale=0.6]{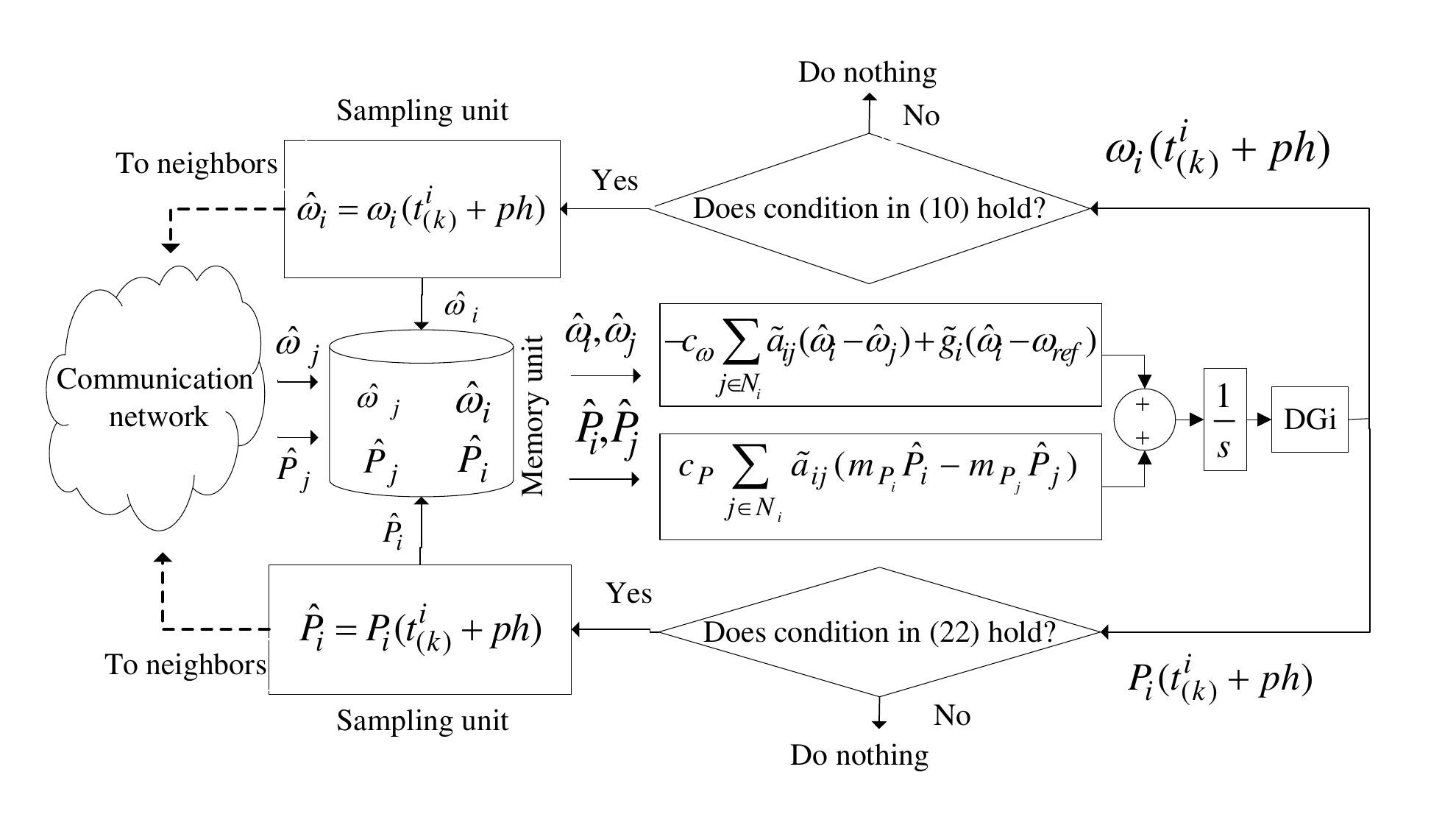}
        \caption{Block diagram of the distributed secondary control with asynchronous periodic event-triggered communication mechanism.}
        \label{Secondary_control}
    \end{figure*}
\begin{remark}
\normalfont From practical standpoint, one should consider the effect of $\lambda$ and its restrictions on the size of processors sampling periods as well as the number of events. The term $\lambda$ itself is directly dependent on the controller gain $c_\omega$, since we multiplied the graph Laplacian matrix by the controller gain $c_w$ along the proof. Hence, determining appropriate $\lambda$ is a trade-off problem between the speed of convergence and the computation complexity.
\end{remark}

\section{Case Studies}

    In this section, an islanded ac microgrid test system is developed in MATLAB/Simulink environment to demonstrate the performance of the proposed event-based active power and frequency controllers through evaluating three scenarios. In the first case, we check our proposed secondary control ability to restore frequency deviations caused by droop controller and accurate power sharing. In the second case, we will make a comparison between our asynchronous event-based method with the proposed method discussed in \cite{bidram2013secondary}. Robustness of the proposed control scheme against load changes is evaluated in the last case. Here we consider a 380 V, 50 Hz microgrid system consisting of four DGs with a strongly connected communication graph $\mathcal{G}$ as shown in Fig.~\ref{Communication}. The inner loop control parameters and load specifications are provided in Table \ref{T1} and \ref{T2}.  Let the graph Laplacian of the communication topology be 
    
        \[
        \tilde{L}=
          \begin{bmatrix}
            1 & 0 & 0 & -1 \\
            -1 & 1 & 0 & 0 \\
            0 & -1 & 1 & 0 \\
            0 & 0 & -1 & 1
          \end{bmatrix}.
        \]
    
    \begin{figure}[t]
        \centering
        \includegraphics[scale=0.7]{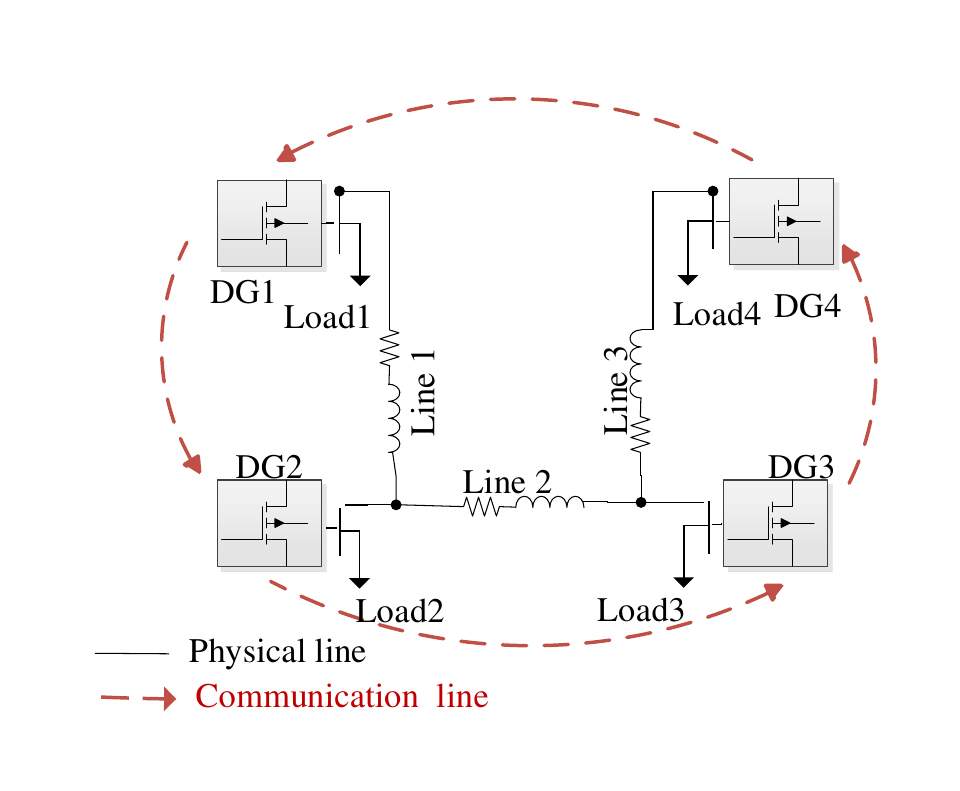}
        \caption{Single-line diagram of the studies microgrid system.}
        \label{Communication}
    \end{figure}
       \begin{table}[]
    \caption{Specification of the microgrid system.}
    \centering
    \begin{tabular}{|c|c|c|c|c|c|}
    \hline
    \multicolumn{6}{|c|}{\textbf{DGs}}                                                                                                                                                                                    \\ \hline
    \multicolumn{2}{|c|}{}       & \multicolumn{2}{c|}{\begin{tabular}[c]{@{}c@{}}DG 1 and 2 \\ \end{tabular}} & \multicolumn{2}{c|}{\begin{tabular}[c]{@{}c@{}}DG  3 and 4\\  \end{tabular}} \\ \hline
    \multicolumn{2}{|c|}{$mP$}     & \multicolumn{2}{c|}{$9.4 \times 10^{-5}$}                                                              & \multicolumn{2}{c|}{$12.5 \times 10^{-5}$}                                                              \\ \hline
    \multicolumn{2}{|c|}{$nQ$}     & \multicolumn{2}{c|}{$1.3\times 10^{-3}$}                                                               & \multicolumn{2}{c|}{$1.4\times 10^{-3}$}                                                                \\ \hline
    \multicolumn{2}{|c|}{$Rc$}     & \multicolumn{2}{c|}{$0.03~\Omega$}                                                               & \multicolumn{2}{c|}{$0.03~\Omega$}                                                                \\ \hline
    \multicolumn{2}{|c|}{$lC$}     & \multicolumn{2}{c|}{0.35 mH}                                                              & \multicolumn{2}{c|}{0.35 mH}                                                               \\ \hline
    \multicolumn{2}{|c|}{$Rf$}     & \multicolumn{2}{c|}{$0.1~\Omega$}                                                                & \multicolumn{2}{c|}{$0.1~\Omega$}                                                                 \\ \hline
    \multicolumn{2}{|c|}{$Lf$}     & \multicolumn{2}{c|}{1.35 mH}                                                              & \multicolumn{2}{c|}{1.35 mH}                                                               \\ \hline
    \multicolumn{2}{|c|}{$Cf$}     & \multicolumn{2}{c|}{0.050 mF}                                                                & \multicolumn{2}{c|}{0.050 mF}                                                                 \\ \hline
    \multicolumn{2}{|c|}{$KPV$}    & \multicolumn{2}{c|}{0.1}                                                                  & \multicolumn{2}{c|}{0.05}                                                                  \\ \hline
    \multicolumn{2}{|c|}{$KIV$}    & \multicolumn{2}{c|}{420}                                                                  & \multicolumn{2}{c|}{390}                                                                   \\ \hline
    \multicolumn{2}{|c|}{$KPC$}    & \multicolumn{2}{c|}{15}                                                                   & \multicolumn{2}{c|}{10.5}                                                                  \\ \hline
    \multicolumn{2}{|c|}{$KIC$}    & \multicolumn{2}{c|}{20000}                                                                & \multicolumn{2}{c|}{16000}                                                                 \\ \hline
    \multicolumn{6}{|c|}{\textbf{Lines}}                                                                                                                                                                                  \\ \hline
    \multicolumn{2}{|c|}{Line 1} & \multicolumn{2}{c|}{Line 2}                                                               & \multicolumn{2}{c|}{Line 3}                                                                \\ \hline
    $Rl1$          & $0.23~\Omega$        & $Rl2$                                       & $0.35~\Omega$                                        & $Rl3$                                        & $0.23~\Omega$                                        \\ \hline
    $Ll1$          & 0.318  mH        & $Ll2$                                       & 1.847 mH                                       & $Ll3$                                        & 0.318 mH                                        \\ \hline
    \end{tabular} \label{T1}
    \end{table}
        \begin{table}[h]
    \caption{Loads per phase of the microgrid system}
    \begin{tabular}{|c|c|c|c|c|c|c|c|}
    \hline
    \multicolumn{2}{|c|}{Load 1} & \multicolumn{2}{c|}{Load 2} & \multicolumn{2}{c|}{Load 3} & \multicolumn{2}{c|}{Load 4} \\ \hline
    $R1$         & $20~\Omega$            & $R2$        & $35~\Omega$           & $R3$        & $35~\Omega$           & $R4$        & $20~\Omega$           \\ \hline
    $L1$         & 0.035 H        & $L2$        & 0.050 H        & $L3$        & 0.050 H        & $L4$        & 0.040 H        \\ \hline
    \end{tabular} \label{T2}
    \end{table}
    In addition, DG 1 is the only DG that can access the reference with the pining gain of $\tilde{g}=1$. Set the controller gains $c_\omega=c_p=4.5$. Then, multiplying the Laplacian graph by these gains, from \eqref{lamda} we obtain $\lambda=9$ for the graph Laplacian associated with the frequency controller. Setting $\sigma_\omega=\sigma_P=0.1$, according to condition \eqref{Eq16}, all DGs are guaranteed to be restored if the sampling period is chosen sufficiently small to satisfy \eqref{Eq16}. We now pick a set of random time instants as $t_0^1=0$ s, $t_0^2=0.005$ s, $t_0^3=0.008$ s, $t_0^4=0.009$ s and set the event-checking period $h=0.01$ s, which satisfies the condition \eqref{Eq16}. A detailed description of the proposed asynchronous periodic event-based secondary control scheme is presented in Algorithm 1.
    
    \begin{algorithm}
    \caption{The proposed event-triggered distributed secondary control Algorithm.}\label{Times} %\hline

    \textbf{Step 1} Initialize $p=0$, $k=0$, $t_{(0)}^i=t_0^i$ \vspace{.01in}\\    
    \textbf{Step 2} Sample and store $y_i(t_{(0)}^i )=y_i(t_0^i)$ and send it to neighbors\vspace{.01in}\\
	\textbf{Step 3} Loop:\vspace{.01in}\\
    - Check the event-triggering mechanisms, \eqref{Proposed C} and \eqref{eq32}\vspace{.01in}\\
    \textbf{If} both event-triggering conditions \eqref{Proposed C} and \eqref{eq32} hold, \vspace{.01in}\\
    -\textbf{then}: \vspace{.01in}\\ 
	- Update and broadcast $\hat{y}_i (t)=y_i (t_{(k)}^i+ph)$,\vspace{.01in}\\
	- $t_{(k)}^i=t_{(k)}^i+ph,k=k+1,p=0$,\\
	-  \textbf{else}: \vspace{.01in}\\
	- broadcast the previous $\hat{y}_i(t)$ without updating\vspace{.01in}\\
	- $p=p+1$. \vspace{.01in}\\
	- \textbf{end if} \vspace{.01in}\\
    \textbf{NOTE:} $y_i(t)$ represents each of the operating frequency, $\omega_i$ and active power, $P_i$ of DG $i$. 
    \end{algorithm}
    
    \subsection{Case 1: Frequency Restoration in Microgrid}
    In this subsection, we evaluate the ability of our proposed control method in frequency restoration. The microgrid is assumed to be islanded from the upstream grid at $t=0$ and only the primary controller is activated. As seen in Fig.~\ref{Case1new}(a), after islanding the microgrid, frequency terms of the DGs deviate from their reference values. At $t=2$ s,  the frequency and active power controllers are activated. After applying the secondary controller, the dropped operating frequency terms of DGs are properly restored to their nominal values. Fig.~\ref{Case1new}(b) shows that the control scheme applied restores frequency while sharing active power accurately.
    \begin{figure}[t]
        \centering
        \includegraphics[width=45ex]{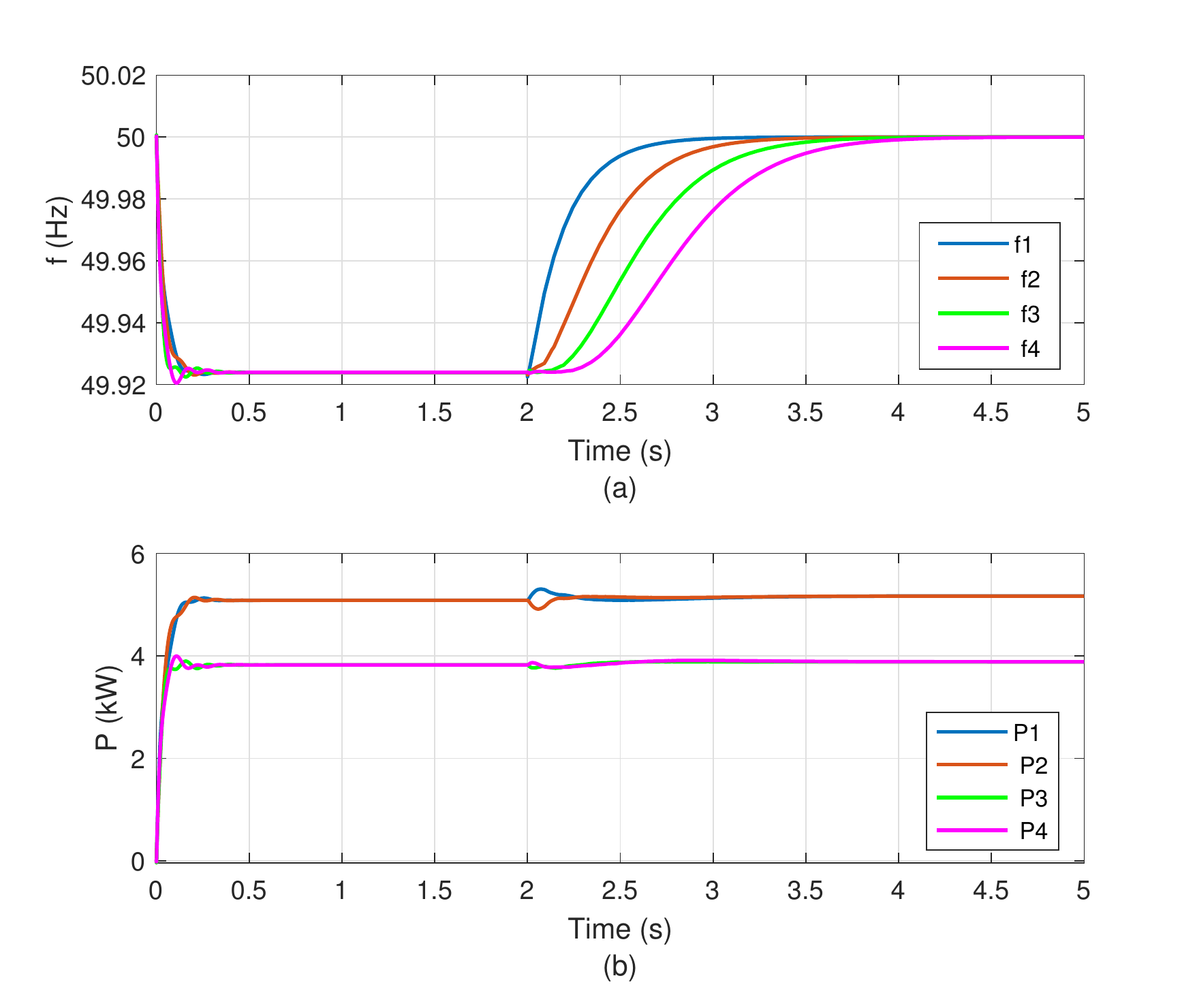}
        \caption{Case 1 DGs' (a) operating frequency; (b) active power.}
        \label{Case1new}
    \end{figure}    
	\subsection{Case 2: Performance Comparisons}
	
    In this subsection, we compare our proposed control method with the conventional time-triggered distributed frequency controller in \cite{bidram2013secondary}, where communications are done synchronously, to demonstrate effectiveness of this integral-type event triggered controller. To this aim, we resimulate Case~1 for both our event-triggered scheme and the proposed method in \cite{bidram2013secondary}. Here we set $c_\omega=c_p=4.5$  and $\tilde{a}_{ij}=1$ for both protocols. It should be noticed that we have considered synchronous periodic communication network in case of time-triggered controller in \cite{bidram2013secondary}. Fig.~\ref{Case3new} shows the performance comparison between our proposed event-triggered method with the controller in \cite{bidram2013secondary}. The outcome underlines that in spite of asynchronous communication, which we considered in our case, the proposed control scheme has an identical performance in comparison with the time-triggered control method in \cite{bidram2013secondary} in terms of frequency restoration and active power sharing.
     
     We now turn our focus to the number of events during the second three simulation period wherein the secondary controller is activated at that time interval (2,5]. The number of communications conducted under our proposed method and the time-triggered communication method in \cite{bidram2013secondary} with the sampling period of $h=0.01$ s are listed in Table~\ref{T3}. It is observed from Table~\ref{T3} that only 19.4\% of the communication lines are busy. These numbers indicate that the proposed integral-type asynchronous periodic event-triggering mechanism effectively reduces the data transmission numbers and computation complexity.
        \begin{figure}[t]
        \centering
        \includegraphics[width=55ex]{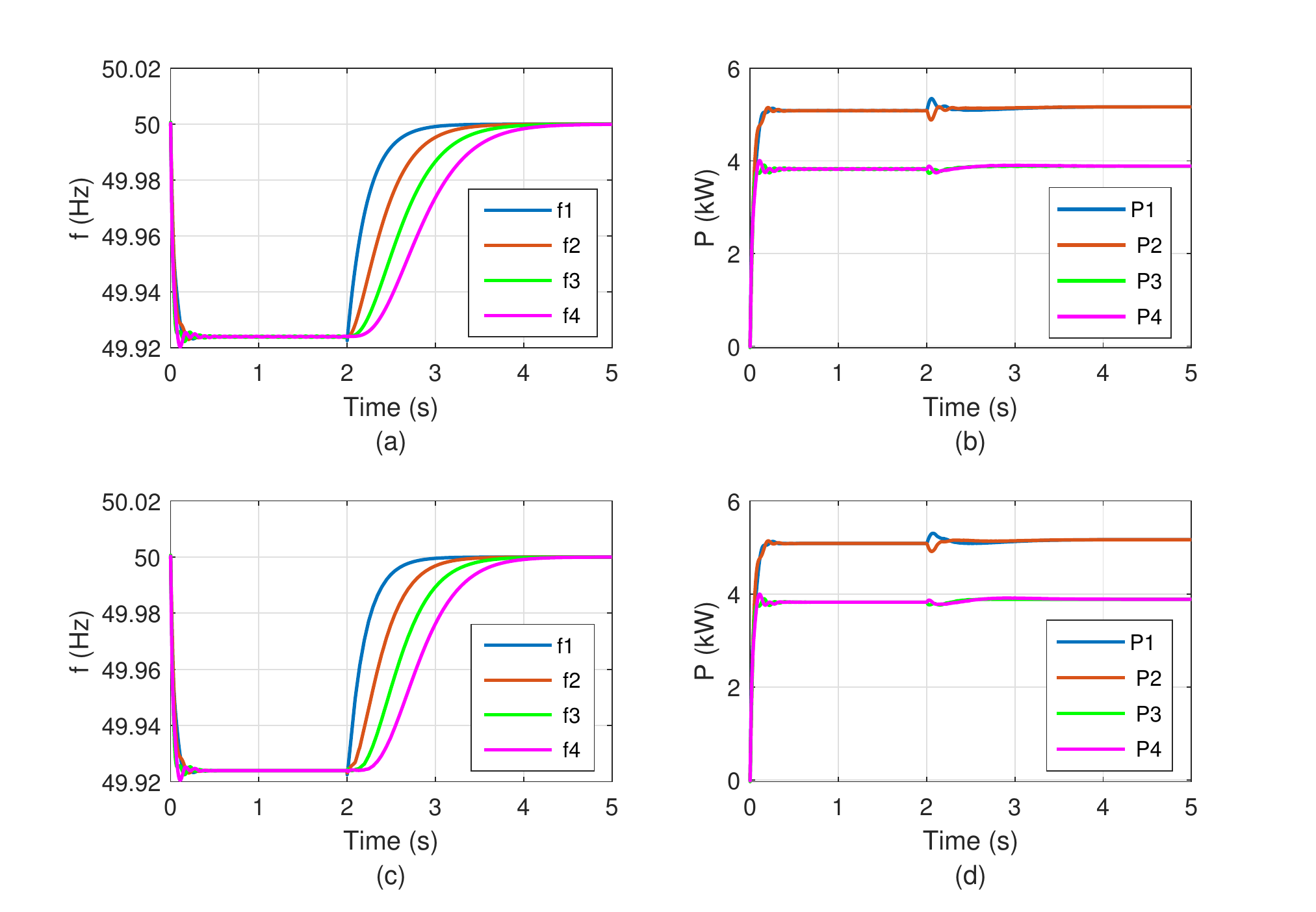}
        \caption{Performance comparison between our proposed control method ((a) and (b)) and that of [6] ((c) and (d)).}% (a) operating frequency; (b) active power magnitudes of the DG units using time-triggered communication; (c) operating frequency; (d) active power magnitudes of the DG units using the proposed event-triggered method.}
        \label{Case3new}
        \end{figure} 
    \begin{table}[h]\centering
    \caption{Communication rate under different data exchange strategies}
    \begin{tabular}{|c|c|c|c|c|c|}
    \hline
     Communication &DG1	&DG2	&DG3&	DG4&	Total  \\ \hline
     Time-Triggered  &	300	&300	&300	&300	&1200 \\ \hline
    Event-Triggered&	34	&60	&68	&71&	233 \\ \hline
    \end{tabular} \label{T3}
    \end{table}
\subsection{Case~3: Performance Analysis Against Load Changes}
    
    In this subsection, robust performance of the proposed controller under load changes is tested. As in the previous cases, it is assumed that the microgrid system is disconnected from the main grid at the beginning and only the primary controller tasks in the first two seconds. At $t=2$ s, the secondary controller is activated, and then, a load inclusion is imposed at $t=5$ s by connecting an $RL$ load with $R=35~\Omega$ and $L=50$ mH in parallel to load 3. To highlight the proposed control method's robust performance, we disconnect the added load at $t=8$ s. Fig.~\ref{Casnew}(a) shows that the proposed secondary control method is able to  remarkably handle these load deviations. Fig.~\ref{Casnew}(b) depicts the capability of the secondary controller in guaranteeing accurate power sharing.
        \begin{figure}[t]
        \centering
        \includegraphics[width=45ex]{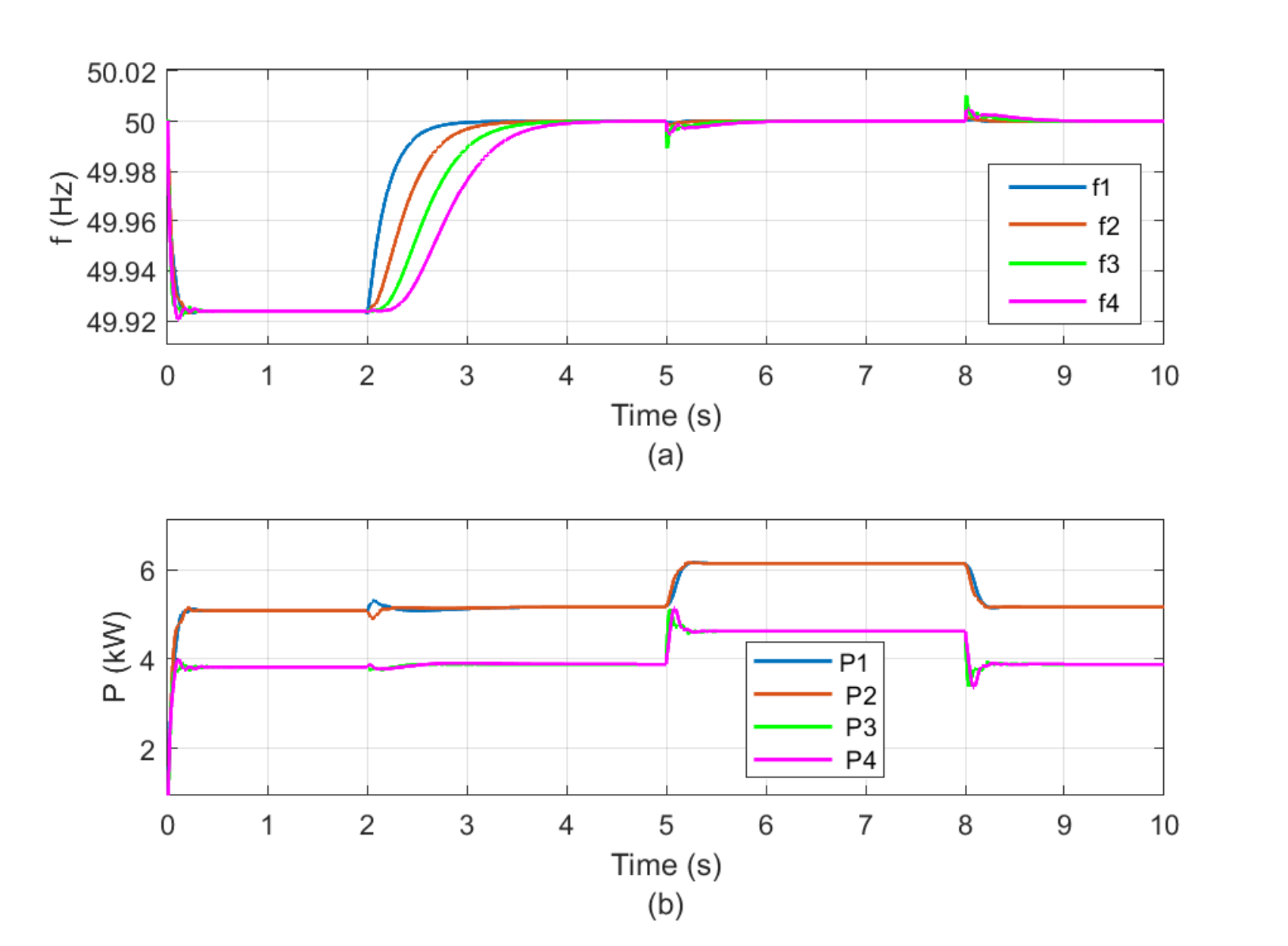}
        \caption{DGs' response under load changes.}
        \label{Casnew}
        \end{figure}    
 \vspace{10pt}
\section{Conclusion}

    In this paper, we have developed a distributed secondary frequency control scheme for an islanded ac microgrid in the case where DGs' clocks are not synchronized. In order to reduce communication and processors computational burden, a sampled-based event-triggered communication mechanism has been developed. In this mechanism, each DG checks its triggering condition periodically according to its own clock, which is possibly asynchronous to those of others. The proposed control scheme eliminates the need for a globally synchronized clock, making it more realistic and practical compared to existing methods. Developing a Lyapunov function, we have obtained a sufficient condition under which the proposed control laws steer all DGs' frequencies to converge to the desired value. Finally, the effectiveness of the proposed method has been verified through simulating a microgrid test-system under different scenarios in MATLAB/Simulink software environment. 
    
    In future work, we will move a step further to design controllers with triggering mechanisms that guarantee the asynchronous restoration problem for microgrid systems, considering communication time delays. Moreover, in addition to asynchronized clocks, the case where DGs have different sampling periods is of great interest.

% References

\balance
%\bibliography{bibliographyIEEETII1}
\bibliography{bibliographyIEEETII2}

\begin{thebibliography}{10}
\providecommand{\url}[1]{#1}
\csname url@samestyle\endcsname
\providecommand{\newblock}{\relax}
\providecommand{\bibinfo}[2]{#2}
\providecommand{\BIBentrySTDinterwordspacing}{\spaceskip=0pt\relax}
\providecommand{\BIBentryALTinterwordstretchfactor}{4}
\providecommand{\BIBentryALTinterwordspacing}{\spaceskip=\fontdimen2\font plus
\BIBentryALTinterwordstretchfactor\fontdimen3\font minus
  \fontdimen4\font\relax}
\providecommand{\BIBforeignlanguage}[2]{{%
\expandafter\ifx\csname l@#1\endcsname\relax
\typeout{** WARNING: IEEEtran.bst: No hyphenation pattern has been}%
\typeout{** loaded for the language `#1'. Using the pattern for}%
\typeout{** the default language instead.}%
\else
\language=\csname l@#1\endcsname
\fi
#2}}
\providecommand{\BIBdecl}{\relax}
\BIBdecl

\bibitem{mousavi2018autonomous}
S.~Y.~M. Mousavi, A.~Jalilian, M.~Savaghebi, and J.~M. Guerrero, ``Autonomous
  control of current-and voltage-controlled dg interface inverters for reactive
  power sharing and harmonics compensation in islanded microgrids,'' \emph{IEEE
  Trans. Power Electronics}, vol.~33, no.~11, pp. 9375--9386, 2018.

\bibitem{lasseter2002microgrids}
R.~H. Lasseter, ``Microgrids,'' in \emph{2002 IEEE Power Engineering Society
  Winter Meeting. Conference Proceedings (Cat. No. 02CH37309)}, vol.~1.\hskip
  1em plus 0.5em minus 0.4em\relax IEEE, 2002, pp. 305--308.

\bibitem{bidram2012hierarchical}
A.~Bidram and A.~Davoudi, ``Hierarchical structure of microgrids control
  system,'' \emph{IEEE Trans. Smart Grid}, vol.~3, no.~4, pp. 1963--1976, 2012.

\bibitem{guerrero2010hierarchical}
J.~M. Guerrero, J.~C. Vasquez, J.~Matas, L.~G. De~Vicu{\~n}a, and M.~Castilla,
  ``Hierarchical control of droop-controlled ac and dc microgrids—a general
  approach toward standardization,'' \emph{IEEE Trans. industrial electronics},
  vol.~58, no.~1, pp. 158--172, 2010.

\bibitem{mehrizi2010potential}
A.~Mehrizi-Sani and R.~Iravani, ``Potential-function based control of a
  microgrid in islanded and grid-connected modes,'' \emph{IEEE Trans. Power
  Systems}, vol.~25, no.~4, pp. 1883--1891, 2010.

\bibitem{bidram2013secondary}
A.~Bidram, A.~Davoudi, F.~L. Lewis, and Z.~Qu, ``Secondary control of
  microgrids based on distributed cooperative control of multi-agent systems,''
  \emph{IET Generation, Transmission \& Distribution}, vol.~7, no.~8, pp.
  822--831, 2013.

\bibitem{qin2016recent}
J.~Qin, Q.~Ma, Y.~Shi, and L.~Wang, ``Recent advances in consensus of
  multi-agent systems: A brief survey,'' \emph{IEEE Trans. Industrial
  Electronics}, vol.~64, no.~6, pp. 4972--4983, 2016.

\bibitem{cao2012overview}
Y.~Cao, W.~Yu, W.~Ren, and G.~Chen, ``An overview of recent progress in the
  study of distributed multi-agent coordination,'' \emph{IEEE Trans. Industrial
  informatics}, vol.~9, no.~1, pp. 427--438, 2012.

\bibitem{gulzar2018multi}
M.~M. Gulzar, S.~T.~H. Rizvi, M.~Y. Javed, U.~Munir, and H.~Asif, ``Multi-agent
  cooperative control consensus: A comparative review,'' \emph{Electronics},
  vol.~7, no.~2, p.~22, 2018.

\bibitem{henningsson2008sporadic}
T.~Henningsson, E.~Johannesson, and A.~Cervin, ``Sporadic event-based control
  of first-order linear stochastic systems,'' \emph{Automatica}, vol.~44,
  no.~11, pp. 2890--2895, 2008.

\bibitem{lunze2010state}
J.~Lunze and D.~Lehmann, ``A state-feedback approach to event-based control,''
  \emph{Automatica}, vol.~46, no.~1, pp. 211--215, 2010.

\bibitem{bidram2014multiobjective}
A.~Bidram, A.~Davoudi, and F.~L. Lewis, ``A multiobjective distributed control
  framework for islanded ac microgrids,'' \emph{IEEE Trans. industrial
  informatics}, vol.~10, no.~3, pp. 1785--1798, 2014.

\bibitem{xu2018optimal}
Y.~Xu, H.~Sun, W.~Gu, Y.~Xu, and Z.~Li, ``Optimal distributed control for
  secondary frequency and voltage regulation in an islanded microgrid,''
  \emph{IEEE Trans. Industrial Informatics}, vol.~15, no.~1, pp. 225--235,
  2018.

\bibitem{dehkordi2018distributed}
N.~M. Dehkordi, H.~R. Baghaee, N.~Sadati, and J.~M. Guerrero, ``Distributed
  noise-resilient secondary voltage and frequency control for islanded
  microgrids,'' \emph{IEEE Trans. Smart Grid}, vol.~10, no.~4, pp. 3780--3790,
  2018.

\bibitem{ahumada2015secondary}
C.~Ahumada, R.~C{\'a}rdenas, D.~Saez, and J.~M. Guerrero, ``Secondary control
  strategies for frequency restoration in islanded microgrids with
  consideration of communication delays,'' \emph{IEEE Trans. Smart Grid},
  vol.~7, no.~3, pp. 1430--1441, 2015.

\bibitem{amoateng2017adaptive}
D.~O. Amoateng, M.~Al~Hosani, M.~S. Elmoursi, K.~Turitsyn, and J.~L. Kirtley,
  ``Adaptive voltage and frequency control of islanded multi-microgrids,''
  \emph{IEEE Trans. Power Systems}, vol.~33, no.~4, pp. 4454--4465, 2017.

\bibitem{ding2018distributed}
L.~Ding, Q.-L. Han, and X.-M. Zhang, ``Distributed secondary control for active
  power sharing and frequency regulation in islanded microgrids using an
  event-triggered communication mechanism,'' \emph{IEEE Trans. Industrial
  Informatics}, vol.~15, no.~7, pp. 3910--3922, 2018.

\bibitem{chen2017secondary}
M.~Chen, X.~Xiao, and J.~M. Guerrero, ``Secondary restoration control of
  islanded microgrids with a decentralized event-triggered strategy,''
  \emph{IEEE Trans. Industrial Informatics}, vol.~14, no.~9, pp. 3870--3880,
  2017.

\bibitem{fan2016distributed}
Y.~Fan, G.~Hu, and M.~Egerstedt, ``Distributed reactive power sharing control
  for microgrids with event-triggered communication,'' \emph{IEEE Trans.
  Control Systems Technology}, vol.~25, no.~1, pp. 118--128, 2016.

\bibitem{zhou2019distributed}
J.~Zhou, Y.~Xu, H.~Sun, L.~Wang, and M.-Y. Chow, ``Distributed event-triggered
  $h_\infty$ consensus based current sharing control of dc microgrids
  considering uncertainties,'' \emph{IEEE Trans. Industrial Informatics}, 2019.

\bibitem{ren2007information}
W.~Ren, R.~W. Beard, and E.~M. Atkins, ``Information consensus in multivehicle
  cooperative control,'' \emph{IEEE Control systems magazine}, vol.~27, no.~2,
  pp. 71--82, 2007.

\bibitem{xiao2009finite}
F.~Xiao, L.~Wang, J.~Chen, and Y.~Gao, ``Finite-time formation control for
  multi-agent systems,'' \emph{Automatica}, vol.~45, no.~11, pp. 2605--2611,
  2009.

\bibitem{pogaku2007modeling}
N.~Pogaku, M.~Prodanovic, and T.~C. Green, ``Modeling, analysis and testing of
  autonomous operation of an inverter-based microgrid,'' \emph{IEEE Trans.
  power electronics}, vol.~22, no.~2, pp. 613--625, 2007.

\bibitem{rasheduzzaman2015reduced}
M.~Rasheduzzaman, J.~A. Mueller, and J.~W. Kimball, ``Reduced-order
  small-signal model of microgrid systems,'' \emph{IEEE Trans. Sustainable
  Energy}, vol.~6, no.~4, pp. 1292--1305, 2015.

\bibitem{wang2017consensus}
A.~Wang, B.~Mu, and Y.~Shi, ``Consensus control for a multi-agent system with
  integral-type event-triggering condition and asynchronous periodic
  detection,'' \emph{IEEE Trans. Industrial Electronics}, vol.~64, no.~7, pp.
  5629--5639, 2017.

\end{thebibliography}
\bibliographystyle{IEEEtran}
\noindent

\end{document}